\documentclass{amsart}

\usepackage{amscd}
\usepackage{amssymb}
\usepackage{amsmath}

\newtheorem{theor}{Theorem}[section]{\bf}{\it}
\newtheorem{lem}[theor]{Lemma}{\bf}{\it}
{\bf}{\it}
\newtheorem{rem}{Remark}[section]{\bf}{\rm}

\numberwithin{equation}{section}
\numberwithin{theor}{section}
\numberwithin{rem}{section}

\def\func#1{\mathop{\rm #1}}

\begin{document}
\title[]{The Classical Inverse Problem for Multi-Particle Densities in 
the Canonical Ensemble Formulation}
\author{}
\address{}
\email{}
\thanks{}
\author{Irina Navrotskaya}
\curraddr{}
\email{}
\thanks{}
\thanks{}

\begin{abstract}
We provide sufficient conditions for the solution of the classical 
inverse problem in the canonical distribution for multi-particle 
densities.  Specifically, we show that there exists a unique potential 
in the form of a sum of $m$-particle ($m\geq 2$) 
interactions producing a given 
$m$-particle density.  The existence and uniqueness of the solution to the multi-particle inverse problem is essential for the numerical simulations of matter using effective potentials derived from structural data.  Such potentials are often employed in coarse-grained modeling.  
The validity of the multi-particle inverse conjecture also has implications for liquid state theory.  For example, it provides the first step in proving the existence of the 
hierarchy of generalized Ornstein-Zernike relations.  
For the grand canonical distribution, the multi-particle inverse problem has been solved by Chayes and Chayes \cite{ChCh84}.  However, the setting of the canonical ensemble presents unique challenges arising from the impossibility of uncoupling interactions when the number of particles is fixed.  
\end{abstract}

\maketitle

\section{Introduction\label{intro}}

Consider a system of $N$ identical particles with coordinates $x_{1},...,x_{N}$ in
some complete $\sigma $-finite measure space $(\Lambda; dx)$ with non-zero measure $dx$.  In the simplest case, $\Lambda \subset \mathbb{R}^{n},n=1,2$, or $3$, $x_{i}$ is the position of
the $i^{th}$ particle, and $dx$ is the Lebesgue measure.  However, the  results here are not limited to this case and apply whenever the conditions stated in the paper are satisfied.  Particularly, 
$(\Lambda; dx)$ may be phase space and may include internal variables.  The completion of the product measure 
$dx^{\otimes k}$ on 
$\Lambda^{k}$ for $1\leq k\leq N$, is denoted by $d^{k}x$.      

We begin with a somewhat formal description of the problem.   
The total potential of the system has the form $W+U$, where 
$W(x_{1},...,x_{N})$ is a fixed scalar internal potential, 
and $U$ is an additional internal or external potential.  Since the particles in the system are identical, $W$ and $U$ are required to be symmetric functions.   

For $1\leq m \leq N$, the classical $m$-particle density in the canonical distribution is defined (up to a multiplicative constant) as \cite{Hansen86}
\begin{equation}
\rho^{(m)}_{U}(x_{1},...,x_{m})=\frac{\int_{\Lambda
^{N-m}}e^{-W-U}dx_{m+1}\cdots dx_{N}}{Z(U)},  \label{1}
\end{equation}
where 
\begin{equation}
Z(U)=\int_{\Lambda ^{N}}e^{-W-U}d^{N}x  \label{2}
\end{equation}
is the canonical partition function.  (The conditions ensuring that 
$0 < Z(U) < \infty $ will be specified later.)  Throughout the paper the inverse temperature $\beta$ is taken to be $1$.  

The $m$-particle inverse problem investigated here is whether, given a symmetric and positive function  $\rho^{(m)}$ on $\Lambda ^{m}$ with $\int_{\Lambda^{m}}\rho^{(m)}d^{m}x=1$, there exists a unique\footnote{The uniqueness of $U$ implies uniqueness of $u$.  See Theorem \ref{th_uniq}. } potential $U$ of the form 
\begin{equation}
U(x_{1},...,x_{N})=\sum_{1\leq i_{1}<\cdots <i_{m}\leq
N}u(x_{i_{1}},...,x_{i_{m}})\quad\text{a.e.},  \label{3} 
\end{equation}
such that $\rho^{(m)}_{U}=\rho^{(m)}$ a.e. on $\Lambda^m$.  
In (\ref{3}), $u$ is a symmetric function on $\Lambda ^{m}$, so $U$ is sought as a sum of symmetric $m$-particle interactions.  A generalization of this problem to systems of several species (mixtures) will be addressed in a separate paper.  

When $m=1$ the inverse problem originates from density functional theory for inhomogeneous fluids \cite{Evans79}.  It was solved for this case 
by Chayes et al. 
\cite[Sections 2 and 8]{ChChLi84}, for both the canonical and grand canonical
ensembles.  When $m\geq 2,$ but \textit{only} in the grand canonical ensemble, it was solved by Chayes and Chayes \cite{ChCh84}
when $\Lambda $ has finite $dx$ measure.   

An immediate application of the multi-particle ($m\geq 2$) 
inverse problem is to the 
numerical modeling of liquid solutions and soft matter, particularly 
coarse-grained (CG) modeling.  In a typical CG simulation, the structural data on CG sites (such as a set of pair or higher-order structural correlation functions 
\cite{Noid13}) are "inverted" to find two- or multi-particle interactions between CG sites.  The existence and uniqueness of such interactions have been implicitly assumed by the long established numerical methods designed for calculating effective potentials from structural data, such as iterative Boltzmann inversion (IBI) \cite{Muller02} and inverse Monte Carlo (IMC) \cite{LyLa95}.  (For the description of basic methods for deducing CG potentials and a survey of literature on the subject, see a recent review by Noid and numerous references therein \cite{Noid13}.)  

However, to the best of our knowledge, the inversion procedure has 
never been justified for the canonical ensemble, even though numerical 
simulations are often performed in this setting.  Instead, it has been erroneously assumed that the conclusions of \cite{ChCh84} are equally valid for the canonical distribution.  
This is a serious misconception that has been persisting in literature for decades.  The setting of the grand canonical ensemble eliminates major difficulties that arise when the number $N$ of particles is fixed.

The existence (and differentiability) of the map between $m$-particle densities and $m$-particle interaction potentials  also has interesting applications to liquid state theory.  For example, it implies the existence of a hierarchy of generalized Ornstein-Zernike (OZ) relations connecting 
$2m$-particle "direct correlation function" and $2m$-,...,$m$-particle densities.  (This will be the subject of a future paper.)

In this article, we provide sufficient conditions for the existence and uniqueness of 
the solutions to the inverse problem in the canonical distribution when $m\leq N$.  These results are summarized in Theorem \ref{Th_ex_un_sol}.  (Note, that for $m=N$, the solution is trivial.  Namely, $U=-\log \rho^{(N)} -W+C$, where $C\in\Bbb{R}$ is any constant.)  
To the best of our knowledge, there is no mathematical treatment of 
the $m\geq 2$ case in that setting.  The uniqueness
-but not existence- was proved by Henderson \cite{He74} when $m=2.$  

Similarly to Chayes et al. \cite{ChCh84,ChChLi84} we use a variational procedure.  Some of the arguments are adopted (after a suitable modification) from \cite[Sections 2 and 6]{ChChLi84}.  However, a number of issues have to be treated differently due to difficulties arising from the coupling of interactions.  In addition to the variational arguments, the verification 
of the multi-particle inverse conjecture for the canonical ensemble involves non-trivial measure-theoretical problems.  These problems were addressed by 
Navrotskaya and Rabier in \cite{Rabier15} where they 
were formulated as measure-theoretical properties of 
generalized $N$-means from $U$-statistics.  We will refer to 
the results established in \cite{Rabier15} when needed.

\section{Preliminaries\label{prelim}}

In this section we introduce some terminology and state the assumptions. 

All functions considered take values in the extended real line $\overline{\Bbb{R}}$, i.e. in $[-\infty,+\infty]$, unless specified otherwise.  The complement of a set $E\subset \Lambda^k$ is denoted $E^c=\Lambda^k\setminus E$, and $\left | E \right | =\int_{E} d^k x$ is the 
$d^{k} x$ measure of a measurable set $E\subset\Lambda^{k}$.  "Almost everywhere" ("a.e.") is always understood relative to the measure $d^k x$, with $k$ obvious from the context, and the same is true regarding the measurability of functions.  Subsets of $\Lambda^k$ of $d^k x$ measure $0$ will be called \textit{null} sets.   Since reference will be very often made to the complements of the 
null sets, we will call such complements \textit{co-null}.  

We shall assume that $\rho^{(m)}$ is the $m$-variable reduction to $\Lambda^{m}$ of some a.e. positive and symmetric probability density $P$ on $\Lambda ^{N}$, i.e. $P>0$ and symmetric a.e., $\int_{\Lambda ^{N}}Pd^{N}x=1$, and 
\begin{equation}
\rho^{(m)}(x_{1},...,x_{m})=\int_{\Lambda
^{N-m}}P(x_{1},...,x_{N})dx_{m+1}\cdots dx_{N}~\text{a.e. on $\Lambda^m$.}  
\label{4}
\end{equation}
It follows from the properties of $P$ that $\rho^{(m)}$ is a.e. positive and symmetric, $\int_{\Lambda^{m}} \rho^{(m)}d^{m}x=1$, and if integration in (\ref{4}) is performed with respect to any set of $N-m$ variables, the function so obtained is 
$\rho^{(m)}$ evaluated at the remaining $m$ variables.  

The fixed internal potential $W$ is assumed to be an a.e. finite and 
symmetric measurable function on $\Lambda^{N}$.  Particularly, 
$| W^{-1}(\infty) |=0$, 
which means that our conditions exclude hard core interactions.  
The absence of hard cores is also a main assumption 
in \cite{ChChLi84}, but they were allowed in \cite{ChCh84}.  

We will also require that 
\begin{equation}
(W+\log P)_{+}\in L^{1}(\Lambda ^{N};Pd^{N}x).  \label{7a}
\end{equation}
The quantity $\int_{\Lambda^N}(W+\log P) P d^{N}x$ is an 
analogue (in the canonical ensemble) of one of the two 
key functionals employed in the density functional 
theory \cite{Evans79}.  Therefore, condition (\ref{7a}) 
assures that this functional is well-defined in some sense.  
On these grounds a similar condition was imposed in \cite{ChCh84}.    
\section{Existence and Uniqueness}
\label{sec:ex_uniq}

Similarly to \cite{ChCh84,ChChLi84}, a variational method is used to prove the existence of solutions to the inverse problem.  Consider the functional    
\begin{equation}
\mathcal{F}_{P}(V):=\frac{e^{-\int_{\Lambda ^{N}}VPd^{N}x}}{Z(V)},  \label{8}
\end{equation}
with $Z$ as in (\ref{2}), defined on a set 
\begin{multline}
\mathcal{V}_{P}:=\{V\in L^{1}(\Lambda ^{N};Pd^{N}x),e^{-V}\in L^{1}(\Lambda
^{N};e^{-W}d^{N}x),  \label{9} \\
V(x_{1},...,x_{N})=\sum_{1\leq i_{1}<\cdots <i_{m}\leq
N}v(x_{i_{1}},...,x_{i_{m}})\text{ a.e. on $\Lambda^{N}$},\\   
\text{where $v$ is a symmetric real-valued measurable function on $\Lambda^m$}\}.
\end{multline}
Note that 
$V\in\mathcal{V}_P$ is a.e. finite and symmetric. 

The appearance of a functional in the form of (\ref{8}) in a 
variational proof of the inverse problem is not accidental.  Disregarding the terms independent 
of $V$, $-\log \mathcal{F}_P(V)$ is the 
Kullback and Leibler mean information for discrimination between two probability distributions \cite{Kullback51}.  As applied to (\ref{8}), the mean information is a measure of overlap between the distributions characterized by probability densities $P$ and $e^{-W-V}/Z(V)$.  This quantity is also often referred to as \textit{relative entropy} \cite{Sh08}.  
Murtola et al. \cite{Murtola09} were the first to notice the connection between the functionals used by 
Chayes et al. in \cite{ChCh84,ChChLi84} and the relative entropy considered by Shell in \cite{Sh08}.  
 
Before applying a variational procedure, we need to ensure that $\mathcal{V}_P\neq \emptyset$.  
If $\exp(-W)\in L^{1}(\Lambda^N; d^{N} x)$, then $\mathcal{V}_{P}\neq \emptyset$ because it contains all constant functions.  (Particularly, $\mathcal{V}_{P}\neq \emptyset$ when $W=0$ and $|\Lambda|<\infty$, which is usually the case in the CG simulations described in the introduction.)    
However, not to be limited by this option, we will assume, more generally, that $\mathcal{V}_{P}\neq \emptyset$.  

In this section we will prove that there exists a unique maximizer (up to an additive constant) of $\mathcal{F}_{P}$ on $\mathcal{V}_{P}$.  Moreover, every such maximizer solves the inverse problem.  This settles the question about the existence of solutions on $\mathcal{V}_{P}$.  However, it is not clear why every solution of the inverse problem on $\mathcal{V}_{P}$ should be a maximizer of $\mathcal{F}_{P}$, and so such a solution may not be unique.  The uniqueness problem is more easily solved on a smaller set 
$\mathcal{V}_{\rho _{(m)}}\subset \mathcal{V}_{P}$ (the inclusion will be proven shortly) defined by 
\begin{multline}
\label{10}
\mathcal{V}_{\rho _{(m)}}:=\{V:e^{-V}\in L^{1}(\Lambda
^{N};e^{-W}d^{N}x),\\
V(x_{1},...,x_{N})= 
\sum_{1\leq i_{1}<\cdots <i_{m}\leq N}v(x_{i_{1}},...,x_{i_{m}})\text{ a.e. on $\Lambda^{N}$},\\
\text{where $v$ is a symmetric real-valued measurable function on $\Lambda^m$, 
$v\in L^{1}(\Lambda ^{m};\rho^{(m)}d^{m}x)$}\}.
\end{multline}
Note that 
by Fubini-Tonelli theorem \cite[Theorem 2.39]{Folland99}:
\begin{multline}
\int_{\Lambda^N} V P d^N x=\sum_{1\leq i_{1}<\cdots <i_{m}\leq N} \int_{\Lambda^m} \rho^{(m)}(x_{i_1},...,x_{i_m})
v(x_{i_1},...,x_{i_m})dx_{i_1}\cdots dx_{i_m}\\= 
\binom{N}{m}\int_{\Lambda^m} \rho^{(m)} v d^m x \in \Bbb{R}.
\label{eq15}
\end{multline}
Thus, $\mathcal{V}_{\rho _{(m)}}\subset \mathcal{V}_{P}$.  It is important to note that this inclusion is proper in general.  A counterexample is provided in \cite{Rabier15}. 

We will see that (assuming $\mathcal{V}_{\rho^{(m)}}\neq \emptyset$) every maximizer of $\mathcal{F}_P$ on $\mathcal{V}_{\rho^{(m)}}$ also solves the inverse problem, and if it exists, this maximizer is unique.  Moreover, every solution of the inverse problem on $\mathcal{V}_{\rho^{(m)}}$ is a maximizer of $\mathcal{F}_P$ on this set.  Thus, the uniqueness of a solution on $\mathcal{V}_{\rho^{(m)}}$ (if one exists) is guaranteed.  However, it seems difficult to prove the existence of maximizers of $\mathcal{F}_P$ on $\mathcal{V}_{\rho^{(m)}}$, and so this set may not contain solutions.  Nevertheless, we will prove that $\mathcal{V}_{\rho^{(m)}}=\mathcal{V}_P$ under an additional assumption about the probability density $P$.    Thus, with this additional assumption, the inverse problem has a unique solution on $\mathcal{V}_{\rho^{(m)}}=\mathcal{V}_P$.  This result is summarized in Theorem \ref{Th_ex_un_sol}.  

The functional $\mathcal{F}_P$ is well-defined on $\mathcal{V}_P$ as a positive real number.  Indeed, the numerator in (\ref{8}) is finite and positive, and $Z(V)<\infty$ by definition of $\mathcal{V}_P$.  Also, $V\in L^{1}(\Lambda^{N};P d^{N}x)$, and so $V$ is a.e. finite (using $P>0$ a.e.)  Thus, $Z(V)>0$ since $W$ is a.e. finite as well.  
\begin{rem}
\rm{We mention here in passing that the requirement of $W$ being finite a.e. is actually implied by our other assumptions.  In fact, if $W=\infty$ on a set of positive measure, then $P>0$ a.e. implies that $(W+\log P)_{+}=\infty$ a.e. on this set, and therefore $(W+\log P)_{+}\notin L^{1}(\Lambda^{N};P d^{N}x)$.  Also, if $V\in\mathcal{V}_P$, then  $V$ is a.e. finite.  Thus, if $W=-\infty$ on a set of positive measure, then $e^{-V}\notin L^{1}(\Lambda^N;e^{-W}d^{N}x)$, and so 
$\mathcal{V}_P=\emptyset$.}
\end{rem}  

\begin{lem}
\label{V_P_Convex}
The set $\mathcal{V}_P$ is convex, and $\log \mathcal{F}_P$ is concave on $\mathcal{V}_P$.  More precisely, 
\begin{equation}
\log \mathcal{F}_P(\lambda V_1+ (1-\lambda)V_0) \geq \lambda 
\log \mathcal{F}_P(V_1)+(1-\lambda)\log \mathcal{F}_P(V_0) 
\label{fp_conc}
\end{equation}
for every $\lambda\in(0,1)$ and every $V_0, V_1 \in \mathcal{V}_P$, with equality 
if and only if $V_1-V_0$ is a constant a.e.  In particular, if $U_0, U_1\in \mathcal{V}_P$ are two maximizers of $\mathcal{F}_P$, then $U_1-U_0$ is a constant a.e.  The same statement holds true with $\mathcal{V}_P$ substituted with $\mathcal{V}_{\rho^{(m)}}$.   
\label{lm5}
\end{lem}
\begin{proof}
Let $V_0, V_1\in \mathcal{V}_P$ and $\lambda\in (0,1)$.  Then $\lambda V_1 +(1-\lambda)V_0$ is in $L^1(\Lambda^N; Pd^N x)$ and has the "sum" structure required for membership in $\mathcal{V}_P$.  By Holder's inequality,
\begin{equation}
\label{Holder}
\int_{\Lambda^N} e^{-\lambda V_1 - (1-\lambda)V_0} e^{-W} d^N x \leq ||e^{-V_1}e^{-W}||_{1}^{\lambda} ||e^{-V_0}e^{-W}||_{1}^{(1-\lambda)}< \infty.
\end{equation}
Thus, $\lambda V_1 +(1-\lambda)V_0 \in \mathcal{V}_P$, and so the set 
$\mathcal{V}_P$ is convex.  

From (\ref{8}), $\log \mathcal{F}_P(V)=-\int_{\Lambda^N} V P d^N x-\log Z(V)$.  Since the first term is linear in $V$, inequality (\ref{fp_conc}) holds if and only if \newline $\log Z (\lambda V_1 + (1-\lambda) V_0 )\leq \lambda \log Z(V_1) + (1-\lambda) \log Z(V_0)$.  But this inequality is equivalent to (\ref{Holder}).  Moreover,  (\ref{Holder}) is an equality if and only if $V_1-V_0$ is a.e. a 
constant \cite[Theorem 6.2]{Folland99}.  (We have used the fact that measures $d^{N}x$ and $e^{-W}d^{N}x$ are absolutely continues with respect to each other.)  The proof for $\mathcal{V}_{\rho^{(m)}}$ is the same. 
\end{proof}

\subsection{Existence and uniqueness of maximizers of $\mathcal{F}_P$ on $\mathcal{V}_P$}
This subsection is devoted to the proof of the existence of maximizers of the functional $\mathcal{F}_P$ on the set $\mathcal{V}_P$, except for a single (crucial) issue resolved in \cite{Rabier15}.  We begin with the proof that $\mathcal{F}_{P}$ is bounded.  
\begin{lem}
\label{lmF_Pbdd} The functional $\mathcal{F}_{P}$ is bounded on $\mathcal{V}_{P}.$ More precisely, $0<\mathcal{F}_{P}(V)
\leq e^{\int_{\Lambda ^{N}}(W+\log P)_{+}Pd^{N}x}$ for every $V\in \mathcal{V}_{P}.$
\end{lem}

\begin{proof}
By (\ref{7a}) and Jensen's inequality for the measure $P d^N x$ \cite[Theorem 3.3]{RudinRCA},
\begin{multline}
e^{-\int_{\Lambda^N} (V+(W+\log P)_{+})P d^N x}\leq \int_{\Lambda^N} e^{-V-(W+\log P)_{+}}P d^N x \\
\leq\int_{\Lambda^N} e^{-V-W-\log P} P d^N x=\int_{\Lambda^N} e^{-V-W} d^N x.
\end{multline}
Thus, $0<\mathcal{F}_{P}(V)\leq e^{\int_{\Lambda ^{N}}(W+\log P)_{+}Pd^{N}x}<\infty$.
\end{proof}
From now on, we set 
\begin{equation}
M:=\sup_{V\in \mathcal{V}_{P}}\mathcal{F}_{P}(V)\in (0,\infty ).  \label{M}
\end{equation}
Let $(V_n)\in \mathcal{V}_P$ be a maximizing sequence for  $\mathcal{F}_{P}$, i.e. $\lim_{n\rightarrow\infty}\mathcal{F}_{P}(V_n)=M$.  (Recall that, by definition of $\mathcal{V}_P$ in (\ref{9}), 
$V_n(x_1,...,x_N)=\sum_{1\leq i_{1}<\cdots <i_{m}\leq
N}v_n(x_{i_{1}},...,x_{i_{m}})$ a.e.,  
where $v_n$ is a symmetric real-valued measurable function on 
$\Lambda^m$.) 
If $V\in \mathcal{V}_{P}$ and $C\in\Bbb{R}$, then \newline 
${V+C=\sum_{1\leq i_1<\cdots<i_m\leq N}\left [v(x_{i_1},...,x_{i_m}) + {\binom{N}{m}}^{-1}C\right ]\in \mathcal{V}_{P}}$, and $\mathcal{F}_{P}(V+C)=\mathcal{F}_{P}(V)$.  Therefore, by adding a suitable constant to each $V_n$, it can be assumed that $Z(V_n)=1$.  Thus ${1=\int_{\Lambda^N} e^{-V_n}   e^{-W} d^N x}$, i.e. $e^{-V_n/2}\in L^2(\Lambda^N; e^{-W} d^N x)$, and 
$|| e^{-V_n/2}||_{2,e^{-W} d^N x}  =1$.  Therefore, by reflexivity of $L^2(\Lambda^N; e^{-W} d^N x)$, there is a subsequence (still denoted by $(V_n)$) and $\Pi \in L^2(\Lambda^N; e^{-W} d^N x)$ such that $(e^{-V_n/2})$ converges weakly to $\Pi$ \cite[Theorem 3.18]{Brezis11}. 

\begin{lem}
\label{Pi_nConL2}
$e^{-V_n/2}\rightarrow \Pi$ in $L^2(\Lambda^N; e^{-W} d^N x)$.  (In particular, $\int_{\Lambda^N} \Pi^2 e^{-W}=1$.)
\end{lem}
\begin{proof}
For convenience, let us define $\Pi_n:=e^{-V_n/2}$.  It suffices to show that ${1=\lim_{n\rightarrow\infty}||\Pi_n ||_{2,e^{-W}d^N x}= ||\Pi ||_{2,e^{-W}d^N x}}$.  The weak convergence implies 
\begin{equation}
1=\lim_{n\rightarrow\infty}||\Pi_n ||_{2,e^{-W}d^N x}\geq ||\Pi ||_{2,e^{-W}d^N x}.
\label{npi_leq1}
\end{equation}
Let $\varepsilon >0$ be fixed.  There is $n_0\in \Bbb{N}$ such that for each $n\geq n_0$,
\begin{equation}
e^{-\int_{\Lambda^N} V_n P d^N x}=\mathcal{F}_P(V_n) > M(1-\varepsilon).
\label{FgM}
\end{equation}
By Mazur's theorem \cite[Theorem 3.13]{RudinFA}, there is a sequence of convex combinations $(\widetilde{\Pi}_n | n\geq n_0)$, i.e.
\begin{equation}
\widetilde{\Pi}_n=\sum_{k=n_0}^n \lambda_k^{(n)}\Pi_k,\quad
\lambda_k^{(n)}\geq 0\quad\forall ~n_0\leq k\leq n,\quad\text{and 
$\sum_{k=n_0}^n \lambda_k^{(n)}=1$},
\end{equation}
such that ${\lim_{n\rightarrow\infty} || \widetilde{\Pi}_n -\Pi ||_{2,e^{-W} d^N x}}=0$.
For every $n\geq n_0$ choose $j_{n},k_{n}\in \{n_0,...,n\}$ such that 
$\langle\Pi _{j_{n}},\Pi _{k_{n}}\rangle\leq \langle\Pi _{j},\Pi _{k}\rangle$ for every $j,k\in
\{n_0,...,n\},$ where $\langle\cdot ,\cdot \rangle$ denotes the inner product of $ L^{2}(\Lambda ^{N};e^{-W}d^{N}x).$  Then, 
\begin{multline}
\label{tildePn_norm}
||\widetilde{\Pi} _{n}||_{2,e^{-W}d^{N}x}^{2}=\\
\sum_{j,k=n_0}^{n}\lambda_j^{(n)}\lambda _{k}^{(n)}\langle\Pi _{j},\Pi _{k}\rangle\geq \left( \sum_{j,k=n_0}^{n}\lambda
_{j}^{(n)}\lambda _{k}^{(n)}\right) \langle\Pi _{j_{n}},\Pi _{k_{n}}\rangle=\langle\Pi _{j_{n}},\Pi
_{k_{n}}\rangle.
\end{multline}
Since $\langle\Pi _{j_{n}},\Pi _{k_{n}}\rangle=\int_{\Lambda ^{N}}e^{-\frac{
V_{j_{n}}+V_{k_{n}}}{2}-W}d^{N}x,$ (\ref{tildePn_norm}) reads
\begin{equation}
||\widetilde{\Pi} _{n}||_{2,e^{-W}d^{N}x}^{2}\geq Z(Y_{n}),  \label{tPiZY}
\end{equation}
where $Y_{n}:=\frac{V_{j_{n}}+V_{k_{n}}}{2}\in \mathcal{V}_P$ (because $\mathcal{V}_P$ is convex by Lemma \ref{lm5}).  Therefore, using (\ref{FgM}), 
\begin{equation}
M\geq \frac{e^{-\int_{\Lambda^N} (V_{j_{n}}+V_{k_{n}})/2}P d^N x}{Z(Y_n)}=
\frac{\left (\mathcal{F}_P(V_{j_{n}})\mathcal{F}_P(V_{k_{n}})\right )^{\frac{1}{2}}}{Z(Y_n)}>\frac{M(1-\varepsilon)}{Z(Y_n)}.
\label{MgM}
\end{equation}
Inequalities (\ref{tPiZY}) and (\ref{MgM}) imply that ${||{\widetilde{\Pi}} _{n}||_{2,e^{-W}d^{N}x}^{2}\geq Z(Y_{n})>1-\varepsilon}$. Therefore, using (\ref{npi_leq1}), ${1\geq ||{\Pi}||_{2,e^{-W}d^{N}x}^{2}=\lim_{n\rightarrow\infty} ||\widetilde{\Pi} _{n}||_{2,e^{-W}d^{N}x}^{2}\geq 1-\varepsilon}$ for any $\varepsilon>0$.  Thus, 
$||{\Pi}||_{2,e^{-W}d^{N}x}||=1$.
\end{proof}
It follows from Lemma \ref{Pi_nConL2} (and the fact that  $d^N x$ and $e^{-W} d^N x$ have the same zero measure 
sets) that there is a subsequence (still denoted $(V_n)$) such that ${e^{-V_{n}/2}\rightarrow\Pi}$ a.e. on $\Lambda^N$ \cite[Theorem 3.12]{RudinRCA}.  Let us define a co-null set 
\begin{multline}
\label{E_def}
\tilde{E}:=\{(x_1,...,x_N)\in\Lambda^N: 
\Pi(x_1,...,x_N)\in \Bbb{R}, \\
e^{-V_{n}(x_1,...,x_N)/2}\rightarrow\Pi(x_1,...,x_N)\}.
\end{multline}
In particular, $\Pi\in [0,\infty)$ on $\tilde{E}$.  Let us also define a function $U$ on $\Lambda^N$ as
\begin{equation}
\label{defU}
U(x):=\left\{\begin{array}{ll}
-2\log \Pi(x) &\text{if $x\in \tilde{E}\cap\{y\in\Lambda^N: \Pi(y)>0$}\},\\
\infty & \text{if $x\in \tilde{E}\cap\{y\in\Lambda^N: \Pi(y)=0$}\},\\
0 & \text{if $x\in \tilde{E}^{c}$}.  
\end{array} \right .
\end{equation}
It is easy to confirm that $U$ is measurable, and $V_n\rightarrow U$ on $\tilde{E}$.  
We will prove that $U\in \mathcal{V}_P$, and $\mathcal{F}_P(U)=M$.  We have
\begin{equation}
1=\int_{\Lambda^N}\Pi^2 e^{-W}d^N x=\int_{\tilde{E}} e^{-U}e^{-W}d^N x=\int_{\Lambda^N} e^{-U}e^{-W}d^N x.
\label{UinVP_1}
\end{equation}    
Therefore, $e^{-U}\in L^1(\Lambda^N; e^{-W} d^N x)$.  It remains to prove that $U\in L^1(\Lambda^N; P d^N x)$, $U(x_1,...,x_N)=\sum_{1\leq i_1<\cdots<i_m\leq N} u(x_{i_1},...,x_{i_m})$ a.e. for some finite and symmetric measurable function $u$ on $\Lambda^m$ (i.e. $U\in \mathcal{V}_P$), and $\mathcal{F}_P(U)=M$.  
\begin{lem}
\label{U_inL^1}
The function $U$ defined in (\ref{defU})  satisfies $U_{-} \in L^1(\Lambda^N; P d^N x)$.
\end{lem}
\begin{proof}
It follows from (\ref{UinVP_1}) that 
\begin{equation}
1=\int_{\Lambda^N} e^{-U-W}d^N x = \int_{\Lambda^N} e^{-U-W-\log P}Pd^N x \geq \int_{\Lambda^N} e^{-U-(W+\log P)_{+}}P d^N x.
\label{UinVP_2}
\end{equation}
The relations $e^t\geq t_{+}$ and $(-t)_{+}=t_{-}$ for $t\in \overline{\Bbb{R}}$ together with (\ref{UinVP_2}) imply that \newline $\int_{\Lambda^N} \left ( U + (W+\log P)_{+}\right )_{-} P d^N x \leq 1$.  Particularly, \newline 
$\left (U + (W+\log P)_{+}\right )_{-} \in L^1(\Lambda^N; P d^N x)$.  Next, using $(t+s)_{-}\leq t_{-}+s_{-}$ for $t\in \overline{\Bbb{R}}$ and $s\in{\Bbb{R}}$, with $t= U + (W+\log P)_{+}$ and $s=-(W+\log P)_{+}$, we obtain $U_{-}\leq \left (U + (W+\log P)_{+}\right )_{-} + 
(W+\log P)_{+}$.  Therefore, by (\ref{7a}) and the above, $U_{-} \in L^1(\Lambda^N; P d^N x)$.
\end{proof}
\begin{lem}
\label{U_inL^1_2}
The function $U$ defined in (\ref{defU})  is in $\mathcal{V}_P$.  Moreover, $\mathcal{F}_P(U)=M$ with $M$ from (\ref{M}).  
\end{lem}
\begin{proof}
Since $U_{-}\in L^1(\Lambda^N;P d^N x)$ by Lemma \ref{U_inL^1}, $U\in L^1(\Lambda^N;P d^N x)$ if and only if $U_{+}\in L^1(\Lambda^N;P d^N x)$.  
This is proved below through an estimate that will also yield $\mathcal{F}_P(U)=M$ (after using Theorem 2.3 in 
\cite{Rabier15} to show that $U\in \mathcal{V}_P$).  

Let $S\subset\Lambda^N$ be the co-null set on which $W\in\Bbb{R}$ and $P\in (0,\infty)$.  For $k,\ell \in \Bbb{N},$ define 
\begin{equation}
U^{k}:=\min (U,k),\qquad \phi^{\ell }:=\chi_{S}\min (W+\log P,\ell ).  \label{UkPhil}
\end{equation}
We have that 
\begin{equation}
\label{Phi^ell_in_L^1}
0\leq W+\log P-\phi^{\ell} \leq (W+ \log P)_{+} \quad \text{on $S$}
\end{equation}
because on this set 
$W+\log P-\phi^{\ell}=0$ when $W +\log P\leq\ell$, and 
$0<W+\log P-\phi^{\ell}=(W+\log P)_{+}-\ell < (W+\log P)_{+}$ 
when $W +\log P>\ell$.   
Particularly, $W+\log P-\phi^{\ell}\in L^1(\Lambda^N; P d^N x )$ by (\ref{7a}).  
Moreover, 
\begin{equation}
\lim_{\ell \rightarrow \infty} \int_{\Lambda^N}(W+\log P-\phi^{\ell})P=0
\label{eqn1}
\end{equation}
by dominated convergence.  Since also $U^{k}=-U_{-}+U_{+}^{k}\in L^1(\Lambda^N; P d^N x )$ by Lemma (\ref{U_inL^1}), 
it follows that 
$(U^{k}-V_n)/2 -W-\log P + \phi^{\ell }\in L^1(\Lambda^N; P d^N x )$.  Therefore, by Jensen's inequality
\begin{multline}
e^{-\int_{\Lambda^N} \frac{V_n}{2} P d^N x}e^{\int_{\Lambda^N}\left (\frac{U^{k}}{2}-W-\log P + \phi^{\ell }\right ) P d^N x}\leq \\
\int_{\Lambda^N} e^{-\left (\frac{(V_n-U^{k})}{2} +W+\log P - \phi^{\ell }\right )} P d^N x
=
\int_{\Lambda^N} e^{-\frac{V_n}{2}}e^{\frac{U^{k}}{2}+\phi ^{\ell }}e^{-W} d ^N x.
\label{eqn2}
\end{multline}
Next, we will show that $e^{\frac{U^{k}}{2}+\phi _{\ell }}\in L^2(\Lambda^N;e^{-W} d^N x)$.  Indeed, $\phi^{\ell }\leq W + \log P$ on $S$ implies 
$e^{\phi^{\ell }-W}\leq P$ a.e., and therefore, $e^{\phi ^{\ell }-W}\in L^1(\Lambda^N; d^N x)$.  Also, $\phi ^{\ell }\leq \ell$, 
$U^k\leq k$ imply $U^k +2 \phi ^{\ell }-W\leq k+\ell + \phi ^{\ell }-W$.  Therefore, $e^{U^k+2\phi ^{\ell }-W}\leq e^{k+\ell}e^{\phi ^{\ell }-W}\in
L^1(\Lambda^N; d^N x)$.  Equivalently, 
\begin{equation}
e^{\frac{U^k}{2}+\phi ^{\ell }} \in L^2(\Lambda^N; e^{-W} d^N x).  
\label{eqn3}
\end{equation}

With $k,\ell$ being held fixed, let $n\rightarrow\infty$.  Then, the leftmost side of (\ref{eqn2}) converges to 
$\sqrt{M}e^{\int_{\Lambda^N}(\frac{U^k}{2}-W-\log P +\phi ^{\ell })P d^N x}$.  By (\ref{eqn3}) and weak convergence, the rightmost side of 
(\ref{eqn2}) converges to 
\begin{equation}
\int_{\Lambda^N} e^{-\frac{(U-U^k)}{2}}e^{\phi ^{\ell }-W} d^N x=\int_{\Lambda^{N}} e^{-\frac{(U-U^k)}{2}}e^{-(W+\log P -\phi ^{\ell })} P d^N x\leq 1, 
\label{eqn4}
\end{equation}
where the last inequality follows by $U-U^k\geq 0$, $W+\log P -\phi _{\ell }\geq 0$, and $\int_{\Lambda^N} P=1$.    This yields
\begin{equation}
\label{eqn5}
\sqrt{M}e^{\int_{\Lambda^N}\frac{U^k}{2} P d^N x} e^{-\int_{\Lambda^N} (W+\log P-\phi _{\ell } )P d^N x}\leq 1.
\end{equation}

With $k$ being held fixed, let $\ell\rightarrow\infty$ in (\ref{eqn5}).  Then, by (\ref{eqn1}), we obtain 
$\sqrt{M}e^{\int_{\Lambda^N}\frac{U^k}{2} P d^N x}\leq 1$.  Equivalently, using $U^k=-U_{-}+U_{+}^k$,
\begin{equation}
M\leq e^{-\int_{\Lambda^N}U^kP d^N x}=e^{-\int_{\Lambda^N}U_{+}^k P d^N x}e^{\int_{\Lambda^N}U_{-} P d^N x}.
\label{eqn6}
\end{equation}
Now, $U_{+}^k$ is increasing to $U_{+}$ pointwise.  Thus, by monotone convergence, \newline
$\lim_{k\rightarrow \infty}\int_{\Lambda^N}U_{+}^k P d^N x=\int_{\Lambda^N}U_{+}P d^N x$.  Taking $k\rightarrow\infty$ in (\ref{eqn6}) results into 
\begin{equation}
0<M\leq e^{-\int_{\Lambda^N}U_{+} P d^N x}e^{\int_{\Lambda^N}U_{-} P d^N x},
\label{eqn7}
\end{equation}
which together with Lemma \ref{U_inL^1} shows that 
$U_{+}$ (and therefore $U$) belongs to \newline $L^1(\Lambda^N; P d^N x)$.  In particular, since $P>0$ a.e., 
$U$ is a.e. finite.  

To finish the proof, we need to show that $U\in \mathcal{V}_P$.  Since $V_n\rightarrow U$ a.e., it follows from Theorem 2.3 in \cite{Rabier15} that there is function 
$u$ on $\Lambda^m$ such that $v_n\rightarrow u$ a.e. 
on $\Lambda^m$.   In particular, $u$ is measurable, and 
(using Lemma 2.1 in \cite{Rabier15}) 
$U(x_1,...,x_N)=\sum_{1\leq i_1<\cdots<i_m\leq N} u(x_{i_1},...,x_{i_m})$ a.e. on $\Lambda^N$.  Moreover, redefining $u$ on a set 
of zero measure, it can be assumed that $u$ is symmetric 
real-valued function on $\Lambda^m$.  Thus, $U$ has the sum structure required for membership 
in $\mathcal{V}_P$, and $U\in\mathcal{V}_P$.  Then, 
(\ref{eqn7})  and (\ref{UinVP_1})  
imply that $\mathcal{F}_P(U)=e^{-\int_{\Lambda^N}U P d^N x}=M$.  
\end{proof}
  
We can now finish the proof of the existence and uniqueness of maximizers of $\mathcal{F}_P$ on $\mathcal{V}_P$.  In fact, all the steps of this proof are already completed, and we essentially just need to cite the previous results.  For convenience, we repeat all the assumptions needed for the validity of the next theorem.  Note that the symmetry of $P$ and $W$ has not been used yet, and will not be needed in the reminder of the proof of the existence and uniqueness of maximizers.   
\begin{theor}
\label{th_ex_un}
Let $\rho^{(m)}$ be the $m$-variable reduction of an a.e. positive probability density $P$, i.e. $\rho^{(m)}$ is given by (\ref{4}), and let $W$ be an a.e. finite measurable potential on $\Lambda^N$.  Suppose that 
$(W+\log{P})_{+}\in L^{1}(\Lambda^N; Pd^N x)$.  If $\mathcal{V}_P\neq \emptyset$, there is $U\in\mathcal{V}_P$ (with 
$U(x_1,...,x_N)=\sum_{1\leq i_1<\cdots<i_m\leq N} u(x_{i_1},...,x_{i_m})$ a.e. ) such that 
$\mathcal{F}_P(V)\leq\mathcal{F}_P(U)$ for every $V\in\mathcal{V}_P$.  
Furthermore, if $U_1$ and $U_0$ are two maximizers of $\mathcal{F}_P$ on $\mathcal{V}_P$ (or on $\mathcal{V}_{\rho^{(m)}}$, if they exist), then $U_1-U_0$ is a constant a.e., and the same is true for $u_1-u_0$. 
\end{theor}  
\begin{proof}
By Lemma \ref{U_inL^1_2}, there is $U\in\mathcal{V}_P$ such that 
\begin{equation}
\mathcal{F}_P(V)\leq\mathcal{F}_P(U)~\text{for every $V\in\mathcal{V}_P$.} 
\label{FPV_L_FPU}
\end{equation}

Suppose, there are $U_0$, $U_1\in\mathcal{V}_P$ (or $\mathcal{V}_{\rho^{(m)}}$) satisfying (\ref{FPV_L_FPU}).  Then, by Lemma \ref{lm5}, there is $C\in\Bbb{R}$ such that 
$U_1-U_0=C$ a.e. on $\Lambda^{N}$.  Therefore, by Corollary 2.4 in \cite{Rabier15}, $u_1-u_0={\binom{N}{m}}^{-1}C$ a.e. on $\Lambda^m$.  
\end{proof}
\subsection{Maximizers of $\mathcal{V}_P$ and solutions of the inverse problem}
\label{sec_max_sol}
Now we are ready to prove the existence of solutions to the inverse problem.  
\begin{theor}
\label{th_max_sol}
Let $P$ be an a.e. positive and symmetric probability density, and $W$ be measurable and a.e. finite and symmetric potential on $\Lambda^N$.  Suppose that for some $U\in\mathcal{V}_P$, $\mathcal{F}_P(V)\leq\mathcal{F}_P(U)$ for every $V\in\mathcal{V}_P$.  Then, $\rho_U^{(m)}=\rho^{(m)}$ a.e.  
\end{theor}
\begin{rem}
\rm{According to Theorem \ref{th_max_sol}, every maximizer of $\mathcal{F}_P(V)$ on $\mathcal{V}_P$ is a solution.  Since such maximizers exist by Theorem \ref{th_ex_un}, the solutions of the inverse problem exist under the assumptions of Theorem \ref{th_ex_un}, with the additional requirement for $P$ and $W$ to be a.e. symmetric.}
\end{rem}
\begin{proof}[of Theorem \ref{th_max_sol}]
Let $\xi\in L^{\infty}(\Lambda^m; d^m x)$ be a symmetric real-valued function.  
For every $(x_1,...,x_N)\in\Lambda^N$, set 
\begin{equation}
\Xi (x_{1},...,x_{N}):=\sum_{1\leq i_{1}<\cdots <i_{m}\leq N}\xi
(x_{i_{1}},...,x_{i_{m}}).  \label{Xi}
\end{equation}
By definition of the norm of $L^{\infty }(\Lambda ^{m};d^{m}x)$, there is a
co-null set $T_{m}$ of $\Lambda ^{m}$ such that $|\xi (x_{1},...,x_{m})|\leq
||\xi ||_{\infty,d^{m}x}$ for every $(x_{1},...x_{m})\in T_{m}$.  By Lemma 2.1 in \cite{Rabier15}, 
$| \Xi (x_{1},...,x_{N})| \leq\binom{N}{m} ||\xi ||_{\infty,d^{m}x}$ for every $(x_1,...,x_N)$ in some co-null set $T_N\subset\Lambda^N$, and so $\Xi\in L^{\infty }(\Lambda ^{N};d^{N}x)$.  

For every $t\in\Bbb{R}$, $U+t\Xi\in\mathcal{V}_P$, and therefore, $\mathcal{F}_P(U+t\Xi)\leq\mathcal{F}_P(U)$.  
The function $t\mapsto\mathcal{F}_P(U+t\Xi)$ is smooth.  This is obvious for the numerator in (\ref{8}).  The same property for the denominator follows by a theorem on 
differentiation of parameter-dependent integrals \cite[Theorem 2.27]{Folland99}.   ($\Xi\in L^{\infty }(\Lambda ^{N};d^{N}x)$ is used here.)  
Thus, $\frac{d}{dt}\mathcal{F}_{P}(U+t\Xi )_{|t=0}=0.$  This gives 
\begin{equation}
\label{eq_3.2_1}
\mathcal{F}_P(U)\left[ \frac{\int_{\Lambda^N} e^{-U-W} \Xi d^{N}x}{Z(U)}-
\int_{\Lambda^N} \Xi P d^{N} x   \right ]=0.
\end{equation}    
By symmetry, (\ref{eq_3.2_1}) amounts to 
$\int_{\Lambda^{m}} \xi (\rho_U^{(m)}-\rho^{(m)} ) d^{m} x=0$.  Further, the set $S_m\subset\Lambda^m$ on which 
$\rho^{(m)}$ and $\rho_U^{(m)}$ are both finite and symmetric is co-null.  Choosing 
$\xi=\chi_{S_m}\times\func{sign}(\rho_U^{(m)}-\rho^{(m)} )$ 
(with $\func{sign}(0):=0 $), we obtain 
$\int_{\Lambda^{m}} |\rho_U^{(m)}-\rho^{(m)} | d^{m} x=0$, 
and so 
$\rho_U^{(m)}=\rho^{(m)}$ a.e. on $\Lambda^m$.  
\end{proof}
\begin{rem}
\rm{Theorem \ref{th_max_sol} is still true if $\mathcal{V}_P$ is replaced by $\mathcal{V}_{\rho^{(m)}}$, as can be easily verified by simply repeating the arguments in the proof.}
\end{rem}

\subsection{Uniqueness of the solutions to the inverse problem}
It remains to resolve the problem of the uniqueness of solutions.  So far, we have shown that solutions exist on $\mathcal{V}_P$.  However, as was mentioned in the introduction, they may not be unique.  It turns out though, 
that, under an additional assumption about $P$, the sets $\mathcal{V}_P$ and $\mathcal{V}_{\rho^{(m)}}$ coincide.  Thus, with this assumption, $\mathcal{V}_{\rho^{(m)}}$ contains solutions.  
  According to the next theorem, every solution on $\mathcal{V}_{\rho^{(m)}}$ is a maximizer of $\mathcal{F}_P$, and therefore is unique up to an additive constant by Theorem \ref{th_ex_un}.   
\begin{theor}
\label{th_uniq}
Suppose that $U\in\mathcal{V}_{\rho^{(m)}}$, where \newline $U(x_1,...,x_N)=
\sum_{1\leq i_1<...<i_m\leq N} u(x_{i_1},...,x_{i_m})$ a.e.  
Suppose also that $\rho_U^{(m)}=\rho^{(m)}$ a.e.  Then, $\mathcal{F}_P(V)\leq\mathcal{F}_P(U)$ for every $V\in\mathcal{V}_{\rho^{(m)}}$.  
Consequently, if $U_0, U_1\in \mathcal{V}_{\rho^{(m)}}$, 
and $\rho_{U_0}^{(m)}=\rho_{U_1}^{(m)}=\rho^{(m)}$ a.e., 
then $U_1-U_0$ is a.e. a constant, and the same is true for 
$u_1-u_0$.   
\end{theor}
\begin{proof}
The "Consequently" part of the conclusion follows directly from Theorem \ref{th_ex_un}.  To prove the rest, we need the following lemma:  
\begin{lem}
\label{uniq_lemma}
Let us call the functions in the form of (\ref{Xi}) "admissible," for brevity. 
Then, $\mathcal{F}_P(U+\Xi)\leq\mathcal{F}_P(U)$ for every admissible $\Xi$.  
\end{lem}
\begin{proof}
Since, $\rho_U^{(m)}=\rho^{(m)}$ a.e., 
$\int_{\Lambda^m} \xi (\rho_U^{(m)}-\rho^{(m)})d^m x=0$.  Thus, $\frac{d}{dt}\mathcal{F}_{P}(U+t\Xi )_{|t=0}=0$ (by reversing the steps in the proof of Theorem \ref{th_max_sol}).  Moreover, by Lemma \ref{V_P_Convex}, 
$t\mapsto \log \mathcal{F}_P(U+t\Xi)$ is a concave (smooth) function on $\Bbb{R}$.  Therefore, 
$\mathcal{F}_P(U+t\Xi)$ attains its maximum at $t=0$.  Taking $t=1$, gives 
$\mathcal{F}_P(U+\Xi)\leq\mathcal{F}_P(U)$.
\end{proof}

Next, we will show that for every $V\in V_{\rho^{(m)}}$ and every $\varepsilon>0$ there is an admissible $\Xi$ such that $\mathcal{F}_P(V)\leq\mathcal{F}_P(U+\Xi)+\varepsilon$, and so Lemma \ref{uniq_lemma} implies that $\mathcal{F}_P(V)\leq\mathcal{F}_P(U)$ for every $V\in V_{\rho^{(m)}}$.  

Let $\varepsilon>0$, and $V\in V_{\rho^{(m)}}$, where $V(x_1,...,x_N)=\sum_{1\leq i_1<\cdots<i_m\leq N} v(x_{i_1},...,x_{i_m})$ a.e.  Define $d_n=\max(v-u,-n)$, so $d_n$ is symmetric, bounded below, and $d_n+u\rightarrow v$ pointwise  and in 
$L^{1}(\Lambda^{m};\rho^{(m)} d^m x)$ (by dominated convergence, because $|d_n+u|\leq |v-u|+|u|$).  Let  \newline 
$\Delta_n(x_1,...,x_N)=\sum_{1\leq i_1<\cdots<i_m\leq N} d_n(x_{i_1},...,x_{i_m})$.  Then, 
\begin{multline}
\int_{\Lambda^N} (U+\Delta_n) P d^N x= \binom{N}{m} \int_{\Lambda^m} (u+d_n) \rho^{(m)} d^{m} x\\
\rightarrow \binom{N}{m} \int_{\Lambda^m} v \rho^{(m)} d^{m} x=\int_{\Lambda^N} V P d^{N} x .
\label{uniq_one}
\end{multline}    
Since $e^{-U-\Delta_n}$ is increasing to $e^{-V}$ a.e., the monotone convergence gives
\begin{equation}
\int_{\Lambda^N}e^{-U-\Delta_n-W}d^{N}x\rightarrow 
\int_{\Lambda^N}e^{-V-W}d^{N}x.
\label{uniq_two}
\end{equation}
By (\ref{uniq_one}), (\ref{uniq_two}) and (\ref{8}), there is $n_0\in\Bbb{N}$ such that 
$|\mathcal{F}_P(V)-\mathcal{F}_P(U+\Delta_{n_0})|<\varepsilon/2$.

Define $\xi_n=\min(d_{n_0},n)$.  Then, $\xi_n$ is symmetric, bounded, and 
$u+\xi_n\rightarrow u+d_{n_o}$ pointwise and in $L^{1}(\Lambda^{m};\rho^{(m)} d^m x)$ (by dominated convergence, because $|u+\xi_{n} |\leq |u|+|v-u|$).  Define $\Xi_n(x_1,...,x_N)=\sum_{1\leq i_1<\cdots<i_m\leq N} \xi_n(x_{i_1},...,x_{i_m})$.  Therefore, similarly to (\ref{uniq_one}),
\begin{multline}
\int_{\Lambda^N} (U+\Xi_n) P d^N x= \binom{N}{m} \int_{\Lambda^m} (u+\xi_n) \rho^{(m)} d^{m} x\\
\rightarrow \binom{N}{m} \int_{\Lambda^m}( u+d_{n_0} )\rho^{(m)} d^{m} x=\int_{\Lambda^N} (U+\Delta_{n_0} )P d^{N} x .  
\label{uniq_three}
\end{multline}  
Because $e^{-U-\Xi_n}\rightarrow e^{-U-\Delta_{n_0}}$ a.e., 
and $e^{-U-\Xi_n}\leq e^{-U+\binom{N}{m}n_0}\in L^{1}(\Lambda^N;e^{-W}d^{N} x)$, the dominated convergence yields
\begin{equation}
\int_{\Lambda^N} e^{-U-\Xi_n} d^{N}x\rightarrow \int_{\Lambda^N} e^{-U-\Delta_{n_0}} d^{N}x.  
\label{uniq_four}
\end{equation}
By (\ref{uniq_three}), (\ref{uniq_four}), and (\ref{8}), there is $n_1\in\Bbb{N}$ such that $|\mathcal{F}_P(U+\Delta_{n_0})-\mathcal{F}_P(U+\Xi_{n_1})|<\varepsilon/2$.
Taking $\Xi:=\Xi_{n_1}$, the triangle inequality gives $|\mathcal{F}_P(V)-\mathcal{F}_P(U+\Xi )|<\varepsilon$. 
\end{proof}
\begin{rem}
\rm{The above argument will not work if it is only known that $V\in \mathcal{V}_P$ because $(\Delta_n)_{+}\geq (V - U)_{+}$, and so $U+\Delta_n$ may fail to be in $L^{1}(\Lambda^N;P d^{N}x)$.  If we were to define 
$\Delta_n=\max(V-U,-n)$, then $(\Delta_n)_{+}=(V - U)_{+}$, but $\Delta_n$ would not have the required sum structure to construct $\Xi$.}
\end{rem}
\begin{rem}
\rm{The uniqueness of solutions on $\mathcal{V}_{\rho^{(m)}}$ can also be proved by the argument used by Chayes et al. \cite[Theorem 2.4]{ChChLi84}.  (The same argument 
has also been used by Henderson \cite{He74}.)    
However, Theorem \ref{th_uniq} gives us more information because 
$\mathcal{V}_{\rho^{(m)}}$ may not contain maximizers, and therefore uniqueness of solutions on $\mathcal{V}_{\rho^{(m)}}$ does 
not imply      
that every solution is a maximizer, as Theorem \ref{th_uniq} asserts.}
\end{rem}

\subsection{Existence and uniqueness of the solutions to the inverse problem}
As was mentioned earlier in this section, the inclusion 
$\mathcal{V}_{\rho^{(m)}}\subset\mathcal{V}_{P}$ is proper in general.  Provided that $u\in \mathcal{V}_{P}$, $u\in 
\mathcal{V}_{\rho^{(m)}}$ if and only if $u\in L^{1}(\Lambda^m; \rho^{(m)} d^m x)$.  Theorem 4.5 in \cite{Rabier15} gives us a sufficient condition on $P$ under which $U$ in the form of (\ref{3}) and $U\in L^{1}(\Lambda^N; P d^N x)$ imply that $u\in L^{1}(\Lambda^m; \rho^{(m)} d^m x)$.  Adding this condition to the assumptions made previously, we are able to formulate the following theorem on the existence and uniqueness of solutions to the inverse problem.  
Since the solution for $m=N$ is trivial, we will assume 
that $m<N$.   
\begin{theor}
\label{Th_ex_un_sol}
Let $N\geq 2$, $1\leq m\leq N-1$ be integers, and $\rho^{(m)}$ be an $m$-variable reduction of some a.e. positive and symmetric probability density $P$.  That is, $\rho^{(m)}$ is defined by (\ref{4}).  Suppose that for a.e. $x_N$ in some 
subset $B\subset\Lambda$ of positive $dx$ measure, there 
is a constant $\gamma(x_N)>0$ such that 
\begin{equation}
\label{cond_on_P_sec3}
P(\cdot, x_{N} )\geq \gamma(x_N) \rho^{(N-1)}~ \text{a.e. on $\Lambda^{N-1}$}.  
\end{equation}
Let $W$ be an a.e. finite and symmetric measurable potential on $\Lambda^N$.  Then, $\mathcal{V}_P=\mathcal{V}_{\rho^{(m)}}$.  If, in addition, $\mathcal{V}_P\neq\emptyset$ and $(W+\log{P})_{+}\in L^{1}(\Lambda^N; Pd^N x)$, there is $U\in\mathcal{V}_P$ (with 
$U=\sum_{1\leq i_1<\cdots<i_m\leq N} u(x_{i_1},...,x_{i_m})$ a.e.) 
such that $\rho^{(m)}_{U}=\rho^{(m)}$ a.e., where 
$\rho^{(m)}_{U}$ is defined by (\ref{1}) and (\ref{2}).  Moreover, if $U_0$, $U_1\in\mathcal{V}_P$ and $\rho^{(m)}_{U_0}=\rho^{(m)}_{U_1}=\rho^{(m)}$ a.e., then $U_1-U_0$ is a constant a.e., and the same is true for $u_1-u_0$.
\end{theor}

\begin{rem}
\rm{Condition (\ref{cond_on_P_sec3}) holds for arbitrarily close perturbations in $L^{1}(\Lambda^{N} ; d^{N}x)$ of any 
symmetric probability density $P$.  Moreover, if the 
measure $dx$ is finite, then any bounded symmetric 
probability density can be approximated in 
$L^{\infty}(\Lambda^{N} ; d^{N}x)$ by bounded symmetric 
probability densities satisfying (\ref{cond_on_P_sec3}).  
See Theorems 4.7 and 4.8 in \cite{Rabier15}.}
\end{rem}

\begin{proof}[Proof of Theorem \ref{Th_ex_un_sol}]
According to Theorems \ref{th_ex_un} and \ref{th_max_sol}, there is 
a unique maximizer of $\mathcal{F}_P$ on 
$\mathcal{V}_P$, and this maximizer is a solution of the inverse problem.  Theorem 4.5 in \cite{Rabier15} implies 
that if $P$ satisfies 
(\ref{cond_on_P_sec3}) then 
$\mathcal{V}_P=\mathcal{V}_{\rho^{(m)}}$.  In its turn, 
Theorem \ref{th_uniq} asserts that every solution of the inverse 
problem on $\mathcal{V}_{\rho^{(m)}}$ is a maximizer.  Thus, 
uniqueness of solutions follows.    
\end{proof} 
 
\section*{Acknowledgements}

The author is deeply indebted to Patrick J. Rabier for his generous contributions to this work.  The author is glad to express her gratitude to William Noid for first introducing her to the inverse problem.  


%
%
%

\end{document}